\documentclass{IEEEtran}
\usepackage{cite}
\usepackage{amsmath,amssymb,amsfonts}
\usepackage{algorithmic}
\usepackage{graphicx}
\usepackage{textcomp}

\newtheorem{theorem}{Theorem}
\newtheorem{lemma}{Lemma} 
\newtheorem{assumption}{Assumption}
\newtheorem{proof}{Proof}
\newtheorem{remark}{Remark}%

\def\BibTeX{{\rm B\kern-.05em{\sc i\kern-.025em b}\kern-.08em
    T\kern-.1667em\lower.7ex\hbox{E}\kern-.125emX}}
\begin{document}
\title{A General Model-Based Extended State Observer with Built-In Zero Dynamics}
\author{Jinfeng Chen, Zhiqiang Gao, Yu Hu, and Sally Shao
\thanks{Jinfeng Chen, Zhiqiang Gao and Yu Hu are with the Center for Advanced Control Technologies, Cleveland State University, Cleveland, OH 44115 USA (e-mail: particlefilter2012@gmail.com; z.gao@csuohio.edu; y.hu20@vikes.csuohio.edu). }
\thanks{Sally Shao is a Professor Emeritus in the Department of Mathematics and Statistics, Cleveland
State University, Cleveland, OH 44115 USA (e-mail: s.shao@csuohio.edu).}
}

\maketitle

\begin{abstract}
    A general model-based extended state observer (GMB-ESO) is proposed for 
    single-input single-output linear time-invariant systems with a given 
    state space model, where the total disturbance, a lump sum of model 
    uncertainties and external disturbances, is defined as an extended 
    state in the same manner as in the original formulation of ESO. The 
    conditions for the existence of such an observer, however, are shown 
    for the first time as 1) the original plant is observable; and 2) there 
    is no invariant zero between the plant output and the total disturbance. 
    Then, the finite-step convergence and error characteristics of GMB-ESO 
    are shown by exploiting its inherent connection to the well-known 
    unknown input observer (UIO). Furthermore, it is shown that, with 
    the relative degree of the plant greater than one and the observer 
    eigenvalues all placed at the origin, GMB-ESO produces the identical 
    disturbance estimation as that of UIO. Finally, an improved GMB-ESO 
    with built-in zero dynamics is proposed for those plants with zero 
    dynamics, which is a problem that has not been addressed in all 
    existing ESO designs. 
\end{abstract}


\section{Introduction}
\label{sec:introduction}
Inspired by the war time research and developments, control theory 
(classical and modern) as being taught today worldwide was born largely 
after WWII and grew rapidly during the second half of the last century. 
It focuses on the stability and optimality, premised on the given 
mathematical model of the physical process \cite{aastrom2014control}. 
In the meanwhile, the engineering practice has been dominated by PID 
during all this time, largely model-free. The gap between theory and 
practice has been evident \cite{samad2017survey} and difficult to bridge. 

The crucial gap between theory and practice is the methodology in 
coping with unknown dynamics and unmeasurable disturbances, which 
has been the driving force behind recent developments in active 
disturbance rejection control (ADRC) \cite{han1998active,han2009}. 
Crucial to the success of ADRC is the extended state observer (ESO), 
where all such uncertainties and disturbances are lumped equivalently 
as the total disturbance, treated as an unknown input, defined as 
the extended state, estimated by ESO, and then largely canceled within 
the ESO bandwidth by the control force. In doing so, the plant is 
forced to behave like an ideal integrator chain, for which a 
proportional-derivative (PD) like control law can be used to 
meet most design specifications in practice. As a viable answer 
to the problem of uncertainties and disturbances in a wide range 
of engineering practices \cite{huang2014active,zhang2021overview}, 
ADRC has been added by Mathworks Inc. to its Simulink Control Design 
toolbox \cite{matlabADRC}, next to Model Reference Adaptive Control and PID Auto-tuning. 

The success of ADRC in engineering practice also encourages theoretical 
studies, such as the combination of ESO and proportional-integral (PI) 
controller for better disturbance rejection and transient 
performance \cite{xue2020integrating}; the cascade ESO design to 
reduce the high sensitivity to high-frequency measurement 
noise \cite{lakomy2021cascade}; the series of cascading first-order ESOs 
to overcome peaking and to relax the matching condition of the control 
gain \cite{ran2021new}. Even with such efforts, ESO is still more or 
less a tool of choice by engineers rather than a mature field of academic 
study. Although \cite{li2012generalized} tried to find a systematic method to 
guide ESO design for general systems, the assumption, that the total 
disturbance has a constant value in steady state, is too strong and 
seldom met in practice. The general model-based ESO (GMB-ESO) proposed in 
this paper, however, only assumes no invariant zeros between the output 
and the total disturbance, which makes GMB-ESO more practical for general 
systems as the total disturbance can be arbitrary unknown signal. 

Inspired by the related and well-developed field of research in unknown 
input observer (UIO), this paper is also set to disclose its inherent connection 
to GMB-ESO and the shared insight that ``modeling uncertainties of this kind are 
the so-called unknown inputs which facilitate the determination of a simple 
linear system description'' \cite{hou1998optimal}, albeit maybe not the 
ideal chain of integrators in conventional ESO. The necessary and 
sufficient conditions for the 
existence of UIO have long been rigorously established, including the exact 
decoupling condition and the stable invariant zeros condition \cite{hautus1983}, 
as well as various methods of relaxing them for the sake of practicality 
\cite{ansari2019, sundaram2007, chakrabarty2017delayed, chakrabarty2017state, 
kong2019, ichalal2020}. 
In addition, UIO has also been applied in the studies of fault 
detection/isolation \cite{khan2022design}, decentralized observer 
\cite{hou1994design} and unbiased input and state estimation 
\cite{abooshahab2022simultaneous}. 

The similarity in conceptualization of UIO and ESO is clearly evident in the 
above description, but seldom addressed in the literature. It is of 
great interest to investigate 1) if UIO and ESO can be shown as 
equivalent under some circumstances; 2) if the rigorous analysis of 
UIO can be extended to ESO for better understanding of its error 
dynamics; and 3) if the ESO design can be generalized, for the first 
time, to the cases where the state space model of the plant is fully 
given but with, of course, significant uncertainties and, possibly, 
zero dynamics.

The paper is organized as follows: GMB-ESO is first proposed in Section \ref{sec:GMBESO} 
for a general system with two assumptions. The main theorems illustrating 
its inherent connection to UIO, and their numerical 
simulations are given in Section \ref{sec:connection}. GMB-ESO with built-in zero dynamics 
is introduced in Section \ref{sec:ESOwithzeros}, before 
conclusions are given in Section \ref{sec:conclusion}.

\section{A General Model-Based ESO}
\label{sec:GMBESO}
Consider the following single-input single-output (SISO) linear time-invariant 
discrete-time system with uncertainty
\begin{equation}\label{eq_1}
    \begin{cases}
        x(k+1)=A_0x(k)+B_0u(k)+E_0 f(k)\\
        y(k)=C_0x(k)
    \end{cases}
\end{equation}
where \(x\in\mathbb{R}^n\), \(u\in\mathbb{R}\), and \(y\in\mathbb{R}\) represent 
the state, the known input, and the output, respectively, \(f\in\mathbb{R}\) denotes 
the unknown input and the total disturbance in the UIO and ESO literature, respectively, 
and \(A_0\), \(B_0\), \(E_0\) and \(C_0\) are real and known matrices with 
appropriate dimensions. The following assumptions on system \eqref{eq_1} will be 
adopted throughout this paper.

\begin{assumption}\label{assum1}
    \((A_0, C_0)\) is observable.
\end{assumption}

\begin{assumption}\label{assum2}
    \((A_0, E_0, C_0)\) has no invariant zeros. 
\end{assumption}

Assumption \ref{assum1} is a necessary condition for the existence of an observer. 
Assumptions \ref{assum1} and \ref{assum2} are necessary and sufficient 
conditions for ESO to have the inherent connection with UIO, which will be 
shown in this paper. 

The following Lemma \ref{thm1} gives a property of no invariant zero systems, 
which is used in the proof of Lemma \ref{lem1}. 
Without loss of generality, the known control input \(u\) is ignored in 
all following proofs. 
\begin{lemma}\label{thm1}
    If \((A_0, C_0)\) is observable, then the following statements are equivalent:
    \begin{enumerate}
        \item \(\text{rank}\left(\begin{bmatrix}
        A_0-zI_n & E_0\\
        C_0 & 0
        \end{bmatrix}\right)=n+1\) for all \(z\in \mathbb{C} \), 
        (i.e., the system has no invariant zeros between the disturbance and the output). \label{lemma1a}
        \item  \(C_0E_0 = 0\), \(C_0A_0E_0 = 0\), \(\cdots\), \(C_0A_0^{n-2}E_0=0\), 
        \(C_0A_0^{n-1}E_0\neq 0\). \label{lemma1b}
    \end{enumerate}
\end{lemma}

\begin{proof}
    Since \((A_0, C_0)\) is observable, there exists an invertible matrix \(\bar{S}_1\) 
    such that 
    \begin{equation}\label{thm_eq1}
        \bar{S}_1  \begin{bmatrix}
            A_0 -zI_n & E_0\\
            C_0 & 0
        \end{bmatrix}=
        \begin{bmatrix}
            \bar{A}_0 -zI_n & \bar{E}_0\\
            \bar{C}_0 & 0
        \end{bmatrix}\bar{S}_1,
    \end{equation}
    where \(\bar{A}_0=\begin{bmatrix}
        0 & 1 & \cdots & 0\\
        \vdots & \vdots & \ddots & \vdots\\
        0 & 0 & \cdots & 1 \\
        a_1 & a_2 & \cdots & a_n
    \end{bmatrix}\), 
        \(S_0=\begin{bmatrix}
        C_0\\
        C_0A_0\\
        \vdots\\
        C_0A_0^{n-1}
    \end{bmatrix}\), 
    \(\bar{S}_1=\begin{bmatrix}
        S_0 & 0_{n\times 1}\\
        0_{1\times n} & 1
    \end{bmatrix}\), 
    \(\bar{C}_0=[1, 0, \cdots, 0]_{1\times n}\), \(\bar{E}_0=S_0 E_0\), and
    \([a_1, \cdots, a_n]\triangleq C_0A_0^nS_0^{-1}\). 

    Due to the similarity property in \eqref{thm_eq1}, we have 
        \begin{equation}\label{thm_eq2}
            \begin{array}{r@{}l}
                \text{rank} & \left(\begin{bmatrix}
                    A_0-zI_n & E_0\\
                    C_0 & 0
                \end{bmatrix}\right)=
                \text{rank} \left(\begin{bmatrix}
                    \bar{A}_0-zI_n & \bar{E}_0\\
                    \bar{C}_0 & 0
                \end{bmatrix}\right)\\
                = & \text{rank}\begin{bmatrix}
                    -z & 1 & 0 & \cdots & 0 & C_0E_0\\
                    0 & -z & 1 & \cdots & 0 & C_0A_0E_0\\
                    \vdots & \vdots & \vdots & \ddots & \vdots & \vdots \\
                    0 & 0 & 0 & \cdots & 1 & C_0A_0^{n-2}E_0\\
                    a_{1} & a_{2} & a_{3} & \cdots & a_{n}-z & C_0A_0^{n-1}E_0\\
                    1 & 0 & 0 & \cdots & 0 & 0
                \end{bmatrix}\\
                = & \text{rank}\begin{bmatrix}
                    1 & 0 & \cdots & 0 & C_0E_0\\
                    -z & 1 & \cdots & 0 & C_0A_0E_0\\
                    \vdots & \vdots & \ddots & \vdots & \vdots \\
                    0 & 0 & \cdots & 1 & C_0A_0^{n-2}E_0\\
                    a_{2} & a_{3} & \cdots & a_{n}-z & C_0A_0^{n-1}E_0
                \end{bmatrix} + 1\\
                = & n+1 \quad \text{for all} \ z\in \mathbb{C}. 
            \end{array}
        \end{equation}
    Therefore, the determinant of the last matrix in \eqref{thm_eq2} should be a 
    nonzero constant. All entries including variable \(z\) are in the subdiagonal. 
    By computing the determinant using Laplace Expansion, it is easy to show that the determinant 
    of the last matrix in \eqref{thm_eq2} is a nonzero constant for 
    all \(z\in \mathbb{C}\), i.e., there is no \(z\) in the determinant, 
    if and only if \(C_0E_0=C_0A_0E_0=\cdots=C_0A_0^{n-2}E_0=0\) and 
    \(C_0A_0^{n-1}E_0\neq 0\). \quad \(\square\)
\end{proof}

Note that, from Lemma \ref{thm1}, the system \eqref{eq_1} with no invariant zeros 
between \(f\) and \(y\) satisfies condition \ref{lemma1b}), which is the structure condition 
for the ESO design proposed in \cite{bai2019} for continuous-time systems.

The conventional ESO conceptualizes the plant from input-output view using an 
observability canonical form, whereas, in this paper, from the view of state space model, 
system \eqref{eq_1} with all available model information is considered. 
Note that the observability canonical form in conventional ESO is a 
special case for system \eqref{eq_1}. 
The key to ESO design is the idea that the total disturbance \(f\) is treated as 
an additional state, defined as an extended state. Then ESO can be used to estimate 
the augmented state. The augmented system can be written as 
\begin{equation}\label{eq_2}
    \begin{cases}
        X(k+1) = A X(k) + Bu(k) + E \Delta f(k) \\
        y(k) = C X(k)
    \end{cases}
\end{equation}
where \(X(k)=\begin{bmatrix}
    x(k)\\
    f(k)
\end{bmatrix}\), 
\(
    A = \begin{bmatrix}
    A_0 & E_0\\
    0_{1 \times n} & 1
    \end{bmatrix}
\), 
\(  
    B = \begin{bmatrix}
    B_0\\
    0
    \end{bmatrix}
\), 
\(  C=[C_0, 0]
\), 
\(  E = [0, \cdots, 0, 1]_{1 \times (n+1)}^T
\), and \(\Delta f(k) = f(k+1) - f(k)\). 

\begin{lemma}\label{lem1}
    If Assumptions \ref{assum1} and \ref{assum2} hold, 
    then augmented system \eqref{eq_2} is observable. 
\end{lemma}

\begin{proof}
    Since \((A_0, E_0, C_0)\) has no invariant zeros and \((A_0, C_0)\) is observable, 
    condition \ref{lemma1b}) of Lemma \ref{thm1} holds. Thus, 
    there exists an invertible matrix \(S_1\) such that
    \((A, E, C)\) is equivalent to \((A_1, E_1, C_1)\), i.e., 
    \begin{equation}\label{lem_eq1}
        A_1=S_1 A S_1^{-1}, E_1=S_1 E, C_1=C S_1^{-1},
    \end{equation}
    where 
    \(A_1=\begin{bmatrix}
        \bar{A}_0 & E_0^\prime\\
        0_{1\times n} & 1
    \end{bmatrix}\), 
    \(S_1=\begin{bmatrix}
        S_0 & 0_{n\times 1}\\
        0_{1\times n} & m
    \end{bmatrix}\), \(m=C_0A_0^{n-1}E_0\), 
    \(E_0^\prime=[0, \cdots, 0, 1]^T_{1\times n}\), 
    \(E_1 = [0, \cdots, 0, m]^T_{1\times (n+1)}\), and 
    \(C_1=[1, 0, \cdots, 0]_{1\times (n+1)}\). 

    System \((A_1, E_1, C_1)\) is of an observability canonical form as in 
    conventional ESO, where \((A_1, C_1)\) is observable. Since the observability 
    property is invariant under any equivalent transformation, 
    \((A, C)\) is also observable. \quad \(\square\)
\end{proof}

According to Lemma \ref{lem1}, under assumptions \ref{assum1} and \ref{assum2}, 
a state observer for system \eqref{eq_2} can be designed as 
\begin{equation}\label{eq_3}
    \hat{X}(k+1) = A\hat{X}(k) + Bu(k) + L\left(y(k)-C\hat{X}(k)\right),
\end{equation}
where \(\hat{X}=[\hat{x}, \hat{f}]^T\) is an estimate of the state \(X\), 
and \(L \in \mathbb{R} ^{n+1}\) is the observer gain vector. 
From Lemma \ref{lem1}, the eigenvalues of \(A-LC\) can be arbitrarily 
placed inside the unit circle. For the sake of simplicity, all eigenvalues of 
\(A-LC\) are placed at \(\omega_o\). Note that, for continuous-time systems, 
the eigenvalues are all placed at \(-\omega_o\), where \(\omega_o\) is denoted 
as the observer bandwidth of ESO \cite{gao2003scaling}. 
Different to other ESOs, in this paper, the 
Luenberger state observer \eqref{eq_3} is called GMB-ESO as the 
matrices \(A_0\), \(B_0\), \(C_0\) and \(E_0\) are in general form with 
model information. Moreover, unlike the assumption in \cite{li2012generalized}, 
the total disturbance \(f\) can be arbitrary signal rather than constant 
value in steady state. 

Based on \eqref{eq_2} and \eqref{eq_3}, the error dynamics of 
GMB-ESO \eqref{eq_3} is updated according to the following equation: 
\begin{equation}\label{eq_5}
    e(k+1) = (A-LC)e(k) + E\Delta f(k),
\end{equation}
where \(e(k) = X(k) - \hat{X}(k)\), and \(E\Delta f(k)\) is caused by the 
ignorance of the unknown input \(\Delta f\) in \eqref{eq_3}.

\section{The Connection between GMB-ESO and UIO}
\label{sec:connection}
In this section, the design method of UIO is first introduced, followed by its inherent 
connection to GMB-ESO. Then, numerical simulations are given to validate their 
connection. 

\subsection{Unknown Input Observer}
Similar to ESO, UIO is designed to estimate the state and reconstruct the unknown 
input by using the disturbance decoupling principle. Thus, the dynamics 
of estimation error of UIO is decoupled with the unknown input, 
that is, there is no \(E\Delta f\) in \eqref{eq_5}. 
The necessary and sufficient conditions for the existence of a UIO for 
system \eqref{eq_1} are: 
1) \(\text{rank}(C_0 E_0) = \text{rank}(E_0)\);
2) \(\text{rank} \begin{bmatrix}
    zI_{n} - A_0 & -E_0\\
    C_0 & 0
\end{bmatrix} = n+\text{rank}(E_0)\), \(\forall z \in \mathbb{C} \), 
\(|z| \geqslant 1\) \cite{valcher1999}. 
It is easy to verify that condition 1) is not satisfied for system \eqref{eq_1}, 
but condition 2) is satisfied. Thus, there is no full state real time UIO existed. 
However, a delayed UIO can be constructed \cite{sundaram2007}. 

For the simplicity of studying the connection between GMB-ESO and UIO, the augmented 
system \eqref{eq_2} is used to design a UIO in our simulations because the estimated 
\(\hat{f}\) should be equal to its actual value no matter whether the system 
is augmented or not. 

The full-order state UIO for system \eqref{eq_2} is of the form
\begin{equation}\label{eq_6}
    \hat{X}(k+1)=J\hat{X}(k)+FY[k:k+n+1]+GU[k:k+n],
\end{equation}
where \(Y[k:k+n+1]\) includes all the measurements from \(y(k)\) to \(y(k+n+1)\),
\(U[k:k+n]\) all the inputs from \(u(k)\) to \(u(k+n)\), 
and \(J\), \(F\), and \(G\) are matrices to be determined 
such that the estimation error
\(e(k)=X(k)-\hat{X}(k)\) satisfies
\begin{equation}\label{eq_7}
    e(k+1)=Ne(k),
\end{equation}
and \(N\) is asymptotically stable/nilpotent. 

For the detailed procedure of designing a UIO with delay, 
readers can refer to \cite{sundaram2007}.

\subsection{Finite-Step Convergence}
In this subsection, we are going to show the disturbance \(\hat{f}\) 
estimated by GMB-ESO and UIO are exactly the same mathematically. 
Since the error dynamics of UIO in \eqref{eq_7} is asymptotically 
stable and future measurements are used, the estimated disturbance 
\(\hat{f}\) is equal to the actual disturbance \(f\) with a suitable time delay. 
Reference \cite{ansari2019} has shown that UIO requires a certain number 
of measurements, which are the relative degree of the system plus one 
for a SISO system, to estimate the states. Therefore, the delay of estimation of \(f\) 
is \(n+1\) steps for system \eqref{eq_1}. 

However, the error dynamics of GMB-ESO in \eqref{eq_5} has an input 
term \(\Delta f\). The estimated disturbance \(\hat{f}\) is still equal to 
the actual disturbance \(f\) with a delay of \(n+1\) steps if the system dynamics 
has no invariant zeros between \(f\) and \(y\) and all eigenvalues of the observer 
\eqref{eq_3} are at the origin, which is shown in the 
following Theorem \ref{thm2}. 

\begin{theorem}\label{thm2}
    Condition \eqref{eq_thm1} holds for observer \eqref{eq_3}: 
    \begin{equation}\label{eq_thm1}
        \hat{f}(k)=f(k-n-1)\quad \text{for}\quad k>n+1, 
    \end{equation}
    if and only if (a) system \eqref{eq_1} satisfies Assumptions \ref{assum1} and 
    \ref{assum2}, and (b) all eigenvalues of the observer are placed at the origin. 
\end{theorem}

\begin{proof}
    The estimation error dynamics of observer \eqref{eq_3} for augmented 
    system \eqref{eq_2} is \eqref{eq_5}. 
    Then, by using Lemma \ref{lem1}, \eqref{eq_5} can be written as
    \begin{equation}\label{eq_8}
        e_1(k+1)=(A_1-L_1C_1)e_1(k)+E_1 \Delta f(k),
    \end{equation}
    where \(e_1(k)=S_1e(k)\), \(L_1=S_1L\), and \(A_1\), \(C_1\) and \(E_1\) can 
    be found in \eqref{lem_eq1}. \((A_1, C_1)\)
    can be converted into the observable canonical form with invertible 
    matrix \(S_2\) with
    \begin{equation}\label{eq_9}
        S_2A_1S_2^{-1}=A_2, C_1S_2^{-1}=C_2,
    \end{equation}
    where \(C_2=[1, 0, \cdots, 0]_{1\times (n+1)}\), 
    \(S_2=\begin{bmatrix}
        C_1\\
        C_1A_1\\
        \vdots\\
        C_1A_1^n
    \end{bmatrix}\), and
    \(A_2=\begin{bmatrix}
        0 & 1 & \cdots & 0 & 0\\
        0 & 0 & \cdots & 0 & 0\\
        \vdots & \vdots & \ddots & \vdots & \vdots \\
        0 & 0 & \cdots & 0 & 1\\
        -a_1 & a_1-a_2 & \cdots & a_{n-1}-a_n & a_n+1
    \end{bmatrix}\). 

    Thus, the estimation error in \eqref{eq_8} can be transformed into
    \begin{equation}\label{eq_10}
        e_2(k+1)=(A_2-L_2C_2)e_2(k)+E_2\Delta f(k),
    \end{equation}
    where \(e_2(k)=S_2e_1(k)\), \(E_2=S_2E_1=[0, 0, \cdots, m]_{1\times (n+1)}^T\), 
    and \(L_2=S_2L_1\). 

    Note that matrix \(A_2\) can be transformed into the following observer 
    companion form \(A_3\) by using invertible matrix \(Q_1\) so that we only 
    need to change the entries in the first column of \(A_3\) to place all poles, where 
    \begin{scriptsize}
    \(
        A_3=\begin{bmatrix}
            a_n+1 & 1 & \cdots & 0 & 0\\
            a_{n-1}-a_n & 0 & \cdots & 0 & 0\\
            \vdots & \vdots & \ddots & \vdots & \vdots \\
            a_1-a_2 & 0 & \cdots & 0 & 1\\
            -a_1 & 0 & \cdots & 0 & 0
        \end{bmatrix}
    \), 
    \(
        Q_1=\begin{bmatrix}
            1 & 0 & \cdots & 0 & 0\\
            -1-a_n & 1 & \cdots & 0 & 0\\
            \vdots & \vdots & \ddots & \vdots & \vdots \\
            a_3-a_2 & a_4-a_3 & \cdots & 1 & 0\\
            a_2-a_1 & a_3-a_2 & \cdots & -1-a_n & 1
        \end{bmatrix}
    \). 
    \end{scriptsize}

    By using the invertible matrix \(Q_1\), estimation error dynamics 
    in \eqref{eq_10} can be written as
    \begin{equation}\label{eq_11}
        e_3(k+1)=(A_3-L_3C_3)e_3(k)+E_3\Delta f(k),
    \end{equation}
    where \(e_3(k)=Q_1e_2(k)\), \(A_3=Q_1A_2Q_1^{-1}\), \(L_3=Q_1L_2\), 
    \(C_3=C_2Q_1^{-1}=[1, 0, \cdots, 0]_{1\times (n+1)}\), and 
    \(E_3=Q_1E_2=[0, \cdots, 0, m]_{1\times (n+1)}^T\). 
    Let \(L_3=[l_1, l_2, \cdots, l_{n+1}]^T\). Then 
    \(
        A_3-L_3C_3=\begin{bmatrix}
            a_n+1-l_1 & 1 & \cdots & 0 & 0\\
            a_{n-1}-a_n-l_2 & 0 & \cdots & 0 & 0\\
            \vdots & \vdots & \ddots & \vdots & \vdots \\
            a_1-a_2-l_n & 0 & \cdots & 0 & 1\\
            -a_1-l_{n+1} & 0 & \cdots & 0 & 0
        \end{bmatrix}
    \). 

    Moreover, the characteristic polynomial of matrix 
    \(\begin{bmatrix}
        -\alpha_1 & 1 & 0 & \cdots & 0\\
        -\alpha_2 & 0 & 1 & \cdots & 0\\
        \vdots & \vdots & \vdots & \ddots & \vdots\\
        -\alpha_n & 0 & 0 & \cdots & 1\\
        -\alpha_{n+1} & 0 & 0 & \cdots & 0
    \end{bmatrix}\) is \(z^{n+1}+\alpha_1 z^n+\cdots+\alpha_nz+\alpha_{n+1}\). 
    Therefore, if all eigenvalues are at the origin, the characteristic polynomial 
    is \(z^{n+1}\), which implies \(\alpha_1=\alpha_2=\cdots=\alpha_{n+1}=0\).
    Then, by equating coefficients, we have 
    \[\bar{A}_3=A_3-L_3C_3=\begin{bmatrix}
        0 & 1 & \cdots & 0 & 0\\
        \vdots & \vdots & \ddots & \vdots & \vdots \\
        0 & 0 & \cdots & 1 & 0\\
        0 & 0 & \cdots & 0 & 1\\
        0 & 0 & \cdots & 0 & 0
    \end{bmatrix}. \] 

    After pole placement, \eqref{eq_11} can be written as 
    \begin{equation}\label{eq_12}
        e_3(k+1)=\bar{A}_3e_3(k)+E_3 \Delta f(k). 
    \end{equation}
    Taking the \(z\)-transform of both sides of \eqref{eq_12} without 
    considering initial condition yields
    \begin{equation}\label{eq_13}
        \tilde{e}_3(z)=(zI-\bar{A}_3)^{-1}E_3 \Delta \tilde{f}(z). 
    \end{equation}
    Since all eigenvalues are at the origin, we have
    \begin{equation}\label{eq_14}
        (zI-\bar{A}_3)^{-1}=\begin{bmatrix}
            1/z & 1/z^2 & \cdots & 1/z^{n+1}\\
            \vdots & \vdots & \ddots & \vdots\\
            0 & 0 & \cdots & 1/z^2\\
            0 & 0 & \cdots & 1/z
        \end{bmatrix}.
    \end{equation}

    Next, substituting \(e_1 (k)=S_1 e(k)\), \(e_2 (k)=S_2 e_1 (k)\), 
    \(e_3 (k)=Q_1 e_2 (k)\), \(E_1=S_1 E\), \(E_2=S_2 E_1\), and \(E_3=Q_1 E_2\) 
    into \eqref{eq_13} yields
    \begin{equation}\label{eq_15}
        S_1\tilde{e}(z)=(Q_1S_2)^{-1}(zI-\bar{A}_3)^{-1}(Q_1S_2)S_1E\Delta \tilde{f}(z). 
    \end{equation}
    Due to \((Q_1S_2)S_1E=[0, \cdots, 0, m]_{1\times (n+1)}^T\), \eqref{eq_15} 
    can be written as 
    \begin{equation}\label{eq_16}
        S_1\tilde{e}(z)=(Q_1S_2)^{-1}\begin{bmatrix}
            1/z^{n+1}\\
            1/z^n\\
            \vdots\\
            1/z
        \end{bmatrix}m\Delta\tilde{f}(z).
    \end{equation}

    The sum of all rows of \(Q_1S_2\) is \([0, \cdots, 0, 1]_{1\times (n+1)}\) because 
    \begin{small}
        \begin{equation*}
            Q_1S_2=\begin{bmatrix}
                1 & 0 & \cdots & 0 & 0\\
                -1-a_n & 1 & \cdots & 0 & 0\\
                \vdots & \vdots & \ddots & \vdots & \vdots\\
                a_3-a_2 & a_4-a_3 & \cdots & 1 & 0\\
                a_2-a_1+a_1 & a_3-a_2+a_2 & \cdots & -1-a_n+a_n & 1
            \end{bmatrix}. 
        \end{equation*}
    \end{small}
    Since \((Q_1S_2)^{-1}(Q_1S_2)=I\) and the sum of all rows of \(Q_1S_2\) is the 
    last row of \(I\), the last row of \((Q_1S_2)^{-1}\) is 
    \([1, 1, \cdots, 1]_{1\times (n+1)}\).

    It follows that the last row of \eqref{eq_16} is 
    \begin{equation}\label{eq_17}
        m\left(\tilde{f}(z)-\tilde{\hat{f}}(z)\right)=\begin{bmatrix}
            1 & 1 & \cdots & 1
        \end{bmatrix}\begin{bmatrix}
            1/z^{n+1}\\
            1/z^n\\
            \vdots\\
            1/z
        \end{bmatrix}m \Delta \tilde{f}(z). 
    \end{equation}
    Taking the inverse \(z\)-transform of \eqref{eq_17} yields
    \begin{small}
    \begin{equation}\label{eq_18}
        \begin{array}{r@{}l}
            f(k)-\hat{f}(k) = & \Delta f(k-n-1)+\cdots+\Delta f(k-2)+\Delta f(k-1)\\
            = & f(k-n)-f(k-n-1)+\cdots+f(k-1)\\
            & -f(k-2)+f(k)-f(k-1)\\
            = & f(k)-f(k-n-1). 
        \end{array}
    \end{equation}
    \end{small}
    We have \(\hat{f}(k)=f(k-n-1)\) in \eqref{eq_thm1}. 

    Now we prove if Assumption \ref{assum2} does not hold, then \eqref{eq_thm1} 
    cannot be obtained. Since \((A_0, C_0)\) is observable, there exists an invertible 
    matrix \(S_1\) such that \((A, E, C)\) is equivalent to \((A_1, E_1, C_1)\) like 
    those in Lemma \ref{lem1} except the new \(E_0^\prime=\frac{1}{m}S_0E_0\). 
    From Lemma \ref{thm1}, the new \(E_0^\prime\) has other nonzero entries 
    rather than the only nonzero entry at the last row. We assume that \((A_1, C_1)\) is 
    observable which guarantees \eqref{eq_thm1} to be obtainable. 
    After the same similarity transformations as in the first part of proof, we find 
    that \((Q_1S_2)S_1E\) has other nonzero entries rather than the only nonzero 
    entry at the last row. Moreover, by using the same procedure of deriving the last row 
    of \((Q_1S_2)^{-1}\), we find that all the entries in the last row have identical  
    values. Therefore, the last row of \eqref{eq_15} with invariant zeros 
    can be written as 
    \begin{equation}\label{eq_thm_1add}
        \begin{array}{{r@{}l}}
            m & \left(\tilde{f}(z)- \tilde{\hat{f}}(z)\right) = \gamma [1, 1, \cdots, 1]\\
             & \begin{bmatrix}
                1/z & 1/z^2 & \cdots & 1/z^{n+1}\\
                \vdots & \vdots & \ddots & \vdots\\
                0 & 0 & \cdots & 1/z^2\\
                0 & 0 & \cdots & 1/z
            \end{bmatrix} \begin{bmatrix}
                \beta_1\\
                \beta_2\\
                \vdots \\
                \beta_{n+1}
            \end{bmatrix}\Delta \tilde{f}(z), 
        \end{array}
    \end{equation}
    where \([\beta_1, \cdots, \beta_{n+1}]^T\triangleq (Q_1S_2)S_1E\), and \(\gamma\) 
    is the identical value of the last row of \((Q_1S_2)^{-1}\). From 
    \eqref{eq_thm_1add}, we have \(\hat{f}(k)\) is equal to a linear combination of 
    \(f(k)\), \(f(k-1)\), \(\cdots\), \(f(k-n-1)\) rather than just a single 
    \(f(k-n-1)\). This is a contradiction. \quad \(\square\)
\end{proof}

\begin{remark}
    Since UIO fully decouples the unknown input from the estimation error dynamics, 
    it provides a theoretical limit for the performance of GMB-ESO, which explains that 
    the estimated disturbance \(\hat{f}(k)\) of GMB-ESO is equal to actual \(f(k)\) 
    with a delay of \(n+1\) steps when placing all eigenvalues at the origin. 
\end{remark}

\begin{remark}
    As shown in Theorem \ref{thm2}, Assumptions \ref{assum1} and \ref{assum2} 
    are necessary and sufficient conditions for GMB-ESO to have the inherent connection 
    with UIO.  
\end{remark}

\begin{remark}
    It has been shown that \(\hat{x}\) and \(\hat{f}\) converge to \(x\) and \(f\) 
    exponentially with the increase of magnitude of \(\omega _o\) in continuous 
    time domain for the observability canonical form system with no model information 
    except relative degree in \cite{xue2015} and with model 
    information in \cite{freidovich2008}. However, Theorem \ref{thm2} considers the 
    case that the system is a general SISO linear system with all available model 
    information. 
\end{remark}

\subsection{Error Characteristics of GMB-ESO}
Unique to ESO is the ability to adjust the smoothness of the estimation 
by changing the observer bandwidth \(\omega_o\) according to different level 
of measurement noise. The following Theorem \ref{thm3} shows that 
the estimation accuracy of GMB-ESO can be adjusted by observer bandwidth and 
why it can smooth the estimation under the influence of measurement noise. 
It also gives an accurate estimation error bound of disturbance 
in terms of \(\omega_o\) if the bound of \(\Delta f\) is known. 

\begin{theorem}\label{thm3}
    If Assumptions \ref{assum1} and \ref{assum2} hold for system \eqref{eq_1},
    then 

    (1) The estimation error of \(f\) in observer \eqref{eq_3} is 
    \begin{equation*}
        f(k)-\hat{f}(k)=h(k)\ast \Delta f(k),
    \end{equation*}
    where \(\ast\) represents convolution, \(\Delta f(k) = f(k+1) - f(k)\), 
    and\footnote{For simplicity, let \(\prod_{j=0}^{-1}(k+j)=1\).}
    \begin{small}
        \begin{equation}\label{eq_27}
                \begin{array}{l}
                     h(k)=   \\
                      \begin{cases}
                            1 & 1\leqslant k \leqslant n+1\\
                            \sum\limits_{i = 1}^{n+1}  \frac{1}{(i-1)!}\left(\prod\limits_{j=-i+1}^{-1}(k+j)\right)(1-\omega_o)^{i-1}\omega_o^{k-i} & k\geqslant n+2.
                        \end{cases} 
                \end{array}
        \end{equation}
    \end{small}

    (2) The estimation error of \(f\) decreases 
    monotonically after \(n+1\) steps of the change of \(f\) with reducing the value of 
    \(\omega_o\) \((0\leqslant \omega_o<1)\). 
    
\end{theorem}

\begin{proof}
    After a series of similarity transformations made in Theorem \ref{thm2}, 
    the estimation error dynamics of observer \eqref{eq_3} for 
    the augmented system \eqref{eq_2} can be written as
    \begin{equation}\label{eq_19}
        e_3(k+1)=(A_3-L_3C_3)e_3(k)+E_3 \Delta f(k),
    \end{equation}
    where \(e_3(k)=(Q_1S_2S_1)e(k)\), \(L_3=(Q_1S_2S_1)L\), \(A_3=(Q_1S_2S_1)A(Q_1S_2S_1)^{-1}\), 
    \(C_3=C(Q_1S_2S_1)^{-1}\), and \(E_3=(Q_1S_2S_1)E\). 
    Since \(A_3\) is in observer companion form, the entries of the first 
    column of \(A_3-L_3 C_3\) are determined by the coefficients of 
    characteristic polynomial placed all eigenvalues at \(\omega_o\), which is
    \begin{equation}\label{eq_20}
        A_3-L_3C_3=\begin{bmatrix}
            (-1)^2\binom{n+1}{1}\omega_o & 1 & \cdots & 0 & 0\\
            (-1)^3\binom{n+1}{2}\omega_o^2 & 0 & \cdots & 0 & 0\\
            \vdots & \vdots & \ddots & \vdots & \vdots\\
            (-1)^{n+1}\binom{n+1}{n}\omega_o^n & 0 & \cdots & 0 & 1\\
            (-1)^{n+2}\binom{n+1}{n+1}\omega_o^{n+1} & 0 & \cdots & 0 & 0\\
        \end{bmatrix}. 
    \end{equation}
    Now we need to obtain the transfer function of \eqref{eq_19} by taking \(\Delta f(k)\) as input and
    the last row of \(e(k)\), i.e., \(f(k)-\hat{f}(k)\), as output. According to 
    \eqref{eq_20}, it is not easy to compute \((zI-A_3+L_3C_3)^{-1}\) 
    directly. A similarity transformation \(Q_2\) is used to 
    transform \(A_3-L_3 C_3\) to a controller companion form \(A_4\), where \(A_4=\)
    \begin{small}
    \begin{equation*}
        \begin{bmatrix}
            0 & 1 & \cdots & 0\\
            \vdots & \vdots & \ddots & \vdots\\
            0 & 0 & \cdots & 1\\
            (-1)^{n+2}\binom{n+1}{n+1}\omega_o^{n+1} & (-1)^{n+1}\binom{n+1}{n}\omega_o^{n} & \cdots & (-1)^2\binom{n+1}{1}\omega_o
        \end{bmatrix}, 
    \end{equation*}
    \end{small}
    \begin{equation*}
        Q_2=\begin{bmatrix}
            1 & 0 & \cdots & 0\\
            (-1)\binom{n+1}{1}\omega_o & 1 & \cdots & 0\\
            \vdots & \vdots & \ddots & \vdots\\
            (-1)^{n-1}\binom{n+1}{n-1}\omega_o^{n-1} & (-1)^{n-2}\binom{n+1}{n-2}\omega_o^{n-2} & \cdots & 0\\
            (-1)^{n}\binom{n+1}{n}\omega_o^{n} & (-1)^{n-1}\binom{n+1}{n-1}\omega_o^{n-1} & \cdots & 1 
        \end{bmatrix}. 
    \end{equation*}
    
    Then, \eqref{eq_19} can be reformulated as
    \begin{equation}\label{eq_21}
        e_4(k+1)=A_4e_4(k)+E_4 \Delta f(k),
    \end{equation}
    where \(e_4(k)=Q_2^{-1}e_3(k)\), \(A_4=Q_2^{-1}(A_3-L_3C_3)Q_2\), and 
    \(E_4=Q_2^{-1}E_3\).

    Taking the \(z\)-transform of both sides of \eqref{eq_21} without 
    considering initial condition yields
    \begin{equation}\label{eq_22}
        \tilde{e}_4(z)=(zI-A_4)^{-1}E_4\Delta \tilde{f}(z).
    \end{equation}
    Since \(A_4\) is in a controller companion form and \(E_4=[0, \cdots, 0, m]^T\),
    the transfer function of \eqref{eq_22} is
    \begin{equation}\label{eq_23}
        (zI-A_4)^{-1}E_4=\begin{bmatrix}
            1/(z-\omega_o)^{n+1}\\
            \vdots\\
            z^{n-1}/(z-\omega_o)^{n+1}\\
            z^{n}/(z-\omega_o)^{n+1}
        \end{bmatrix}m. 
    \end{equation}
    From \eqref{eq_19}, \eqref{eq_21}, \eqref{eq_22}, and \eqref{eq_23}, 
    the \(z\)-transform of the original error estimation dynamics is 
    \begin{equation}\label{eq_24}
        S_1\tilde{e}(z)=(Q_1S_2)^{-1}Q_2\begin{bmatrix}
            1/(z-\omega_o)^{n+1}\\
            \vdots\\
            z^{n-1}/(z-\omega_o)^{n+1}\\
            z^{n}/(z-\omega_o)^{n+1}
        \end{bmatrix}m\Delta \tilde{f}(z). 
    \end{equation}
    The last row of \eqref{eq_24} is 
    \begin{equation}\label{eq_25}
        \begin{array}{r@{}l}
            m\left(\tilde{f}(z)-\tilde{\hat{f}}(z)\right) & \\
            = [1, \cdots, 1, 1] & Q_2\begin{bmatrix}
                1/(z-\omega_o)^{n+1}\\
                \vdots\\
                z^{n-1}/(z-\omega_o)^{n+1}\\
                z^{n}/(z-\omega_o)^{n+1}
            \end{bmatrix}m\Delta \tilde{f}(z). 
        \end{array}
    \end{equation}

    Since the original system is linear and satisfies additivity and 
    homogeneity, the estimation error of the disturbance \(f\) decreases with 
    reducing the value of \(\omega_o\) for all \(k\geqslant 0\) if the inverse 
    \(z\)-transform of the transfer function of \eqref{eq_25} 
    decreases with smaller \(\omega_o\) for all \(k\geqslant 0\). The transfer 
    function of \eqref{eq_25} is
    \begin{equation}\label{eq_26}
        \begin{array}{r@{}l}
            \tilde{h}(z) & =[1, \cdots, 1, 1]Q_2\begin{bmatrix}
                1/(z-\omega_o)^{n+1}\\
                \vdots\\
                z^{n-1}/(z-\omega_o)^{n+1}\\
                z^{n}/(z-\omega_o)^{n+1}
            \end{bmatrix}\\
            & =\sum\limits_{i = 1}^{n+1} \dfrac{(1-\omega_o)^{i-1}(z-\omega_o)^{n+1-i}}{(z-\omega_o)^{n+1}}. 
        \end{array}
    \end{equation}
    Taking the inverse \(z\)-transform of \eqref{eq_26} 
    yields \eqref{eq_27}. 

    The derivative of \(h(k)\) with respect to \(\omega_o\) is 
    \begin{equation}\label{eq_28}
        \dfrac{\partial h(k)}{\partial \omega_o} = \dfrac{1}{n!}(1-\omega_o)^n(k-n-1)\left(\prod\limits_{j=-n}^{-1}(k+j)\right)\omega_o^{k-n-2}.
    \end{equation}
    Therefore, if \(0<\omega_o <1\), then \(\partial h/\partial \omega_o >0\) for all \(k>n+1\).
    And \(h(k)\) decreases monotonically with respect to \(\omega_o\) for all 
    \(k>n+1\), which proves (2). \quad \(\square\)
\end{proof}

\begin{remark}
    The monotonic decrease of estimation error of the disturbance \(f\) 
    with respect to \(\omega_o\) shows that the estimation error during the 
    change of \(f\) cannot be increased with a smaller \(\omega_o\). Therefore, 
    Theorem \ref{thm3} guarantees the performance of estimation of \(f\) not only 
    in the steady state but also in the transient state. 
\end{remark}

\begin{remark}
    The accurate estimation error bound of the disturbance \(f\) can be obtained 
    firstly in Theorem \ref{thm3}, whereas \cite{xue2015} and 
    \cite{freidovich2008} provide only conservative error bounds. The accurate 
    error bound can be used in robust control barrier function based quadratic 
    programs for safety critical systems \cite{alan2022disturbance}. It is 
    interesting to find that the estimation error bound of \(f\) is only related to 
    \(\omega_o\) and \(\Delta f\). 
\end{remark}

\subsection{Numerical Validations of the UIO and GMB-ESO Connection}
Consider a position control of a series elastic actuator (SEA) system as in \cite{chen2021}. 
The discretized SEA system with sample time \(20\) ms is given by the matrices
\begin{small}
\(
    A_0 = \begin{bmatrix}
        0 & 1 & 0 & 0\\
        -0.6587 & 1.6494 & 3.4847e-4 & 0\\
        0 & 0 & 0 & 1\\
        1.7991 & 0 & -0.9829 & 1.8929
    \end{bmatrix}, B_0 = \begin{bmatrix}
        0\\
        0\\
        0\\
        74.96
    \end{bmatrix},
\)
\end{small}
\(
    C_0 = [1, 0, 0, 0], E_0 = [0, 0, 0, 1]^T.
\)
Since the relative degree of this system is four, the measurement 
\(y(k+5)\) in \eqref{eq_6} needs to be used to design a UIO, which means 
the UIO has a time delay of five time steps. 

To validate the inherent connection between GMB-ESO and UIO, a step signal of known 
input \(u\) of \(1\) Nm is applied at \(0.1\) s and a step signal of disturbance \(f\) 
provided to the input side of the plant is applied at \(0.5\) s with magnitude of 
\(2.5\) Nm. The same noise with power of \(2\times 10^{-8}\) by using Band-Limited White 
Noise in Simulink$^\circledR$ is added to output measurements in both observers. 
The eigenvalues \(\omega_o\) of GMB-ESO and UIO are all placed at \(1000\) rad/s and \(40\) rad/s 
in continuous time domain (\(2.0612\times 10^{-9}\) and \(0.4493\) in discrete time domain), 
respectively. The initial state of the plant is \(x_0=[0, 0, 0, 0]^T\). 

The actual and estimated total disturbances are given in Fig. \ref{fig1}. The 
total disturbances estimated by GMB-ESO and UIO are the same under influence 
of the exact same white noise, while 
their difference is shown in the bottom. 

\begin{figure}[htbp]
    \begin{center}
        \includegraphics[width=3.5 in]{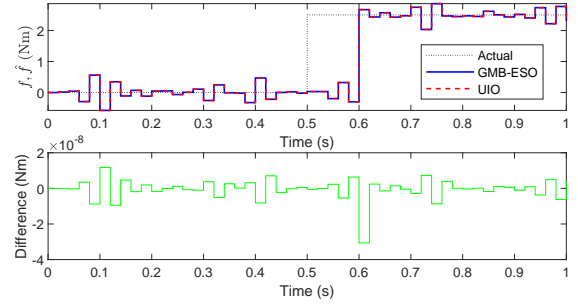}
        \caption{The actual total disturbance and its estimations by GMB-ESO and UIO}
        \label{fig1}
    \end{center}
\end{figure}

Moreover, GMB-ESO can smooth the estimation of the total disturbance. In 
motion control, the sample time is usually \(1\) ms. After discretizing 
the SEA with this sample time, the discretized 
system is given by the matrices
\begin{small}
    \(
        A_0 = \begin{bmatrix}
            0 & 1 & 0 & 0\\
            -0.9793 & 1.9793 & 1.056e-6 & 0\\
            0 & 0 & 0 & 1\\
            0.0046 & 0 & -0.9991 & 1.9989
        \end{bmatrix}, B_0 = \begin{bmatrix}
            0\\
            0\\
            0\\
            0.1904
        \end{bmatrix},
    \)
\end{small}
\(
    C_0 = [1, 0, 0, 0], E_0 = [0, 0, 0, 1]^T.
\)

Since the sample time is very small and the UIO tries to derive a 
very accurate total disturbance to let the plant follow the polluted 
measurements quickly, the estimated total disturbance of UIO is very 
large as shown in the top of Fig. \ref{fig2b}. So, the UIO cannot be used in the case 
when the measurement noise is large and the sample time is small. 
However, by adjusting the observer bandwidth \(\omega_o\), GMB-ESO has the ability 
to smooth the estimated total disturbance. Fig. \ref{fig2b} shows that the 
lower the observer bandwidth is, the smoother the estimated total 
disturbance acts, the slower the estimated total disturbance follows 
the actual total disturbance.

\begin{figure}[htbp]
    \begin{center}
        \includegraphics[width=3.5 in]{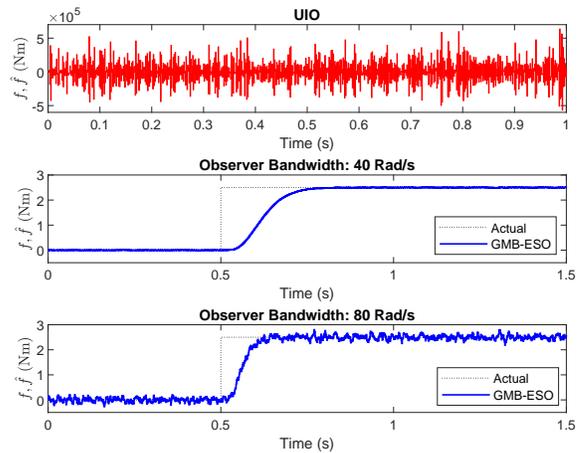}
        \caption{The estimated total disturbances using UIO and GMB-ESO with different observer bandwidth}
        \label{fig2b}
    \end{center}
\end{figure}

\section{GMB-ESO with Built-In Zero Dynamics}
\label{sec:ESOwithzeros}
In this section, GMB-ESO with built-in zero dynamics is proposed for a SISO 
linear system. To make the design method of ESO clearer, the nominal model 
of conventional ESO, i.e., the observability canonical form, is firstly reviewed 
systematically. 

System \eqref{eq_1} considered in the preceding sections only takes into account 
a single total disturbance. However, real-world plants may be subject to 
numerous internal and external disturbances. Only inputs and outputs of the 
plant are available. We need to find a nominal model to represent the actual 
plant as close as possible. 

Consider the following nominal SISO system given by 
\begin{equation}\label{eq_29}
    \begin{cases}
        x(k+1)=A_0 x(k)+B_0u(k)+D_0 d(k)\\
        y(k)=C_0x(k)
    \end{cases}
\end{equation}
where \(d\in \mathbb{R}^q\) is the unknown disturbance vector including internal 
and external disturbances, and \(D_0\in \mathbb{R}^{n\times q}\). Assumption 
\ref{assum1} is satisfied. 
Note that \(E_0f(k)\) in system \eqref{eq_1} is replaced by \(D_0d(k)\) 
as \(f\) represents the total disturbance but \(d\) denotes the actual internal 
and external disturbances. 

\subsection{An Unsolved Problem in ESO Design}
Under Assumption \ref{assum1}, 
if there are no disturbances, a SISO linear system can be transformed to a normal form 
by using the following invertible matrix \cite{isidori1995} 
\begin{equation}\label{eq_31}
    T_1=\begin{bmatrix}
        C_0\\
        \vdots \\
        C_0A_0^{r-1}\\
        \Phi 
    \end{bmatrix},
\end{equation}
where \(\Phi\in\mathbb{R}^{m\times n}\), \(\Phi B_0=0\), and \(r\triangleq n-m\) is the 
relative degree. Since the relative degree is \(r\), we have 
\begin{equation}\label{eq_a1}
    \begin{array}{r@{}l}
        C_0A_0^{i-1}B_0= & 0 \quad \text{for all} \quad i=1, \cdots, r-1\\
        C_0A_0^{r-1}B_0\neq & 0. 
    \end{array}
\end{equation}

Due to the 
existence of disturbances in system \eqref{eq_29}, we set 
\begin{equation}\label{eq_a2}
    \begin{bmatrix}
        \zeta(k) \\
        \eta(k)
    \end{bmatrix}=T_1 x(k)+ \begin{bmatrix} 
        M\\
        0_{m\times (r-1)q}
    \end{bmatrix} \boldsymbol{d}(k),
\end{equation}
where \(\boldsymbol{d}(k)=[d(k)^T, d(k+1)^T, \cdots, d(k+r-2)^T]^T\), 
\(M=\begin{bmatrix}
    0 & 0 & \cdots & 0\\
    C_0D_0 & 0 & \cdots & 0\\
    C_0A_0D_0 & C_0D_0 & \cdots & 0\\
    \vdots & \vdots & \ddots & \vdots \\
    C_0A_0^{r-2}D_0 & C_0A_0^{r-3}D_0 & \cdots & C_0D_0
\end{bmatrix}\), \(\zeta \in \mathbb{R}^r\), 
and \(\eta\in\mathbb{R}^m\). 
For convenience, let \(T_1^{-1}\triangleq [T_{1a}, T_{1b}]\), 
where \(T_{1a}\in \mathbb{R}^{n\times r}\) and \(T_{1b}\in \mathbb{R}^{n\times m}\). 

From \eqref{eq_a2}, we have an expression of \(x\) in terms of \(\zeta\), 
\(\eta\) and \(\boldsymbol{d}\). Substituting this \(x\) and \eqref{eq_a1} into 
system \eqref{eq_29} yields the following normal form 
\begin{equation}\label{eq_30}
    \begin{cases}
        \zeta(k+1)=\hat{A}_0\zeta(k)+\hat{B}_0u(k)+\hat{E}_0(\hat{F}_0\eta(k)+f_1(k))\\
        \eta(k+1)=\hat{S}_0\eta(k)+\hat{G}_0\zeta(k)+d_\eta(k)\\
        y(k)=\hat{C}_0\zeta(k)
    \end{cases}
\end{equation}
where 
\(
    \hat{A}_0=\begin{bmatrix}
        0 & 1 & \cdots & 0\\
        \vdots & \vdots & \ddots & \vdots \\
        0 & 0 & \cdots & 1\\
        \alpha_1 & \alpha_{2} & \cdots & \alpha_r
    \end{bmatrix}
\),
\(  
    \hat{B}_0=\begin{bmatrix}
        0\\
        \vdots \\
        0\\
        b_0
    \end{bmatrix}
\), 
\(  
    \hat{E}_0=\begin{bmatrix}
        0\\
        \vdots \\
        0\\
        1
    \end{bmatrix}
\), 

\([\alpha_1, \cdots, \alpha_r]\triangleq C_0A_0^rT_{1a}\), 
\(b_0=C_0A_0^{r-1}B_0\), 
\(\hat{F}_0 =C_0A_0^rT_{1b}\), 
\(  \hat{C}_0=[1, 0, \cdots, 0]\), \(\hat{S}_0 = \Phi A_0 T_{1b}\), 
\(f_1(k)=\sum\limits_{i = 1}^{r} C_0A_0^{r-i}D_0d(k+i-1) - C_0A_0^rT_{1a}M\boldsymbol{d}(k)\), 
\(\hat{G}_0 = \Phi A_0 T_{1a}\), 
\(d_\eta (k)=\Phi D_0 d(k)- \Phi A_0 T_{1a}M\boldsymbol{d}(k)\). 
The total disturbance in conventional ESO design is defined as 
\begin{equation}\label{eq_total_a}
    f_a(k)=\hat{F}_0\eta(k)+f_1(k). 
\end{equation}

As shown in system \eqref{eq_30}, the zero dynamics is separated from the observability 
canonical form. There are no invariant zeros between \(f_a\) and \(y\), in which \(f_a\) is 
composed of \(\hat{F}_0\eta\) from zero dynamics and \(f_1\) from the  
model uncertainty and external disturbances. Therefore, the information of 
zero dynamics cannot be added into the observer design. 

\subsection{Adding Given Zero Dynamics to GMB-ESO}
To add zero dynamics in ESO design, we set 
\begin{equation}\label{eq_32}
    \bar{x}(k)=T_2 x(k) + P \bar{\boldsymbol{d}}(k),
\end{equation}
where 
\(P=\begin{bmatrix}
    0 & 0 & \cdots & 0\\
    C_0D_0 & 0 & \cdots & 0\\
    C_0A_0D_0 & C_0D_0 & \cdots & 0\\
    \vdots & \vdots & \ddots & \vdots \\
    C_0A_0^{n-2}D_0 & C_0A_0^{n-3}D_0 & \cdots & C_0D_0
\end{bmatrix}\), 
\(T_2=[C_0^T, \cdots, (C_0A_0^{n-1})^T]^T\), 
\(\bar{\boldsymbol{d}}(k)=[d(k)^T, d(k+1)^T, \cdots, d(k+n-2)^T]^T\). 
 
From \eqref{eq_32}, we have an expression of \(x\) in terms of \(\bar{x}\) 
and \(\bar{\boldsymbol{d}}\). Substituting this \(x\) into 
system \eqref{eq_29} yields the following observability canonical form 
\begin{equation}\label{eq_33}
    \begin{cases}
        \bar{x}(k+1)=\bar{A}_0\bar{x}(k)+\bar{B}_0u(k)+\bar{E}_0f_b(k)\\
        y(k)=\bar{C}_0\bar{x}(k)
    \end{cases}
\end{equation}
where \(\bar{A}_0=T_2A_0T_2^{-1}\), 
\(\bar{B}_0=T_2B_0\), \(\bar{C}_0=C_0T_2^{-1}\), \(\bar{E}_0=[0, \cdots, 0, 1]^T\), 
and the total disturbance in GMB-ESO is 
\begin{equation}\label{eq_total_b}
    f_b(k)=\sum\limits_{i = 1}^{n}C_0A_0^{n-i}D_0d(k+i-1) 
    -C_0A_0^nT_2^{-1}P\bar{\boldsymbol{d}}(k).
\end{equation}

System \eqref{eq_33} can be converted back to the original form \eqref{eq_29} with the 
total disturbance \(f_b\) added as follows:
\begin{equation}\label{eq_34}
    \begin{cases}
        \tilde{x}(k+1)=A_0\tilde{x}(k)+B_0u(k)+E_0f_b(k)\\
        y(k)=C_0\tilde{x}(k)
    \end{cases}
\end{equation}
where \(E_0=T_2^{-1}\bar{E}_0\) and \(\tilde{x}=T_2^{-1}\bar{x}\). 

Since there are no invariant zeros between \(y\) and \(f_b\), 
an observer \eqref{eq_3} of augmented system of system \eqref{eq_34} can be 
used to estimate \(\tilde{x}\) and \(f_b\). Comparing \eqref{eq_total_a} and 
\eqref{eq_total_b}, it is clear that the known zero dynamics is no longer
a part of the total disturbance to be estimated, thereby greatly reducing the
load on the observer. 

\begin{remark}
    For multi-output systems, the observability matrix 
    can be easily obtained if the system is 
    observable. Therefore, the method proposed above can be generalized to 
    multi-input multi-output (MIMO) systems. Moreover, ESO-like methods 
    have already been generalized to MIMO systems \cite{wang2015output, wu2020performance}. 
\end{remark}

\section{Concluding Remarks}
\label{sec:conclusion}
To address the problem of uncertainties, the well-established ESO design 
principle is extended in this paper to the plants rigorously studied in 
the existing mathematical control theory: single-input single-output 
linear time-invariant systems with a given state space model. The proposed 
GMB-ESO provides effective means to estimate both the state and the lumped 
disturbance, i.e., the total disturbance, where the data of the combined 
effects of unmodeled dynamics and external disturbances on the plant is 
obtained in real time. Such information, needless to say, is critical in 
allowing the rich body of knowledge in modern control theory to be directly 
applicable to engineering practice where uncertainties abound. 

In particular, the mature and rigorous body of work on UIO functions as a 
bridge in this paper to help analytically establish and numerically verify 
the proposed approach. 
First, Lemma \ref{thm1} shows that the condition for convergence of ESO in 
continuous-time systems \cite{bai2019} is equivalent to Assumption \ref{assum2}, 
i.e., no invariant zeros exist between the plant output and the total disturbance. 
Secondly, Theorem \ref{thm2} shows that UIO sets the ceiling, a theoretical 
limit, for the performance of GMB-ESO, and that GMB-ESO can reach this ceiling only 
if its bandwidth (\(\omega_o\)) is pushed to be infinite and Assumption 
\ref{assum2} holds. In addition, as Theorem \ref{thm3} implies the output 
of GMB-ESO can be made as smooth as needed by adjusting \(\omega_o\) in the 
presence of measurement noise. Furthermore, the closed-form error-bound 
for the estimation of the total disturbance \(f\) is obtained and shown 
to be monotonically decreasing as the bandwidth \(\omega_o\) increases. 
Finally, a particular form of GMB-ESO with built-in zero dynamics is 
introduced for those plants with zeros between the input and output, 
which have been left untreated in all existing ESO designs. 
Including zero dynamics in ESO reduces the load on ESO by incorporating into it the 
known dynamics and not wasting its bandwidth.





\bibliographystyle{IEEEtran}
\bibliography{reference} 




\end{document}